\documentclass[a4paper]{lipics-v2021}

\usepackage{graphicx,amsthm,amsmath,amssymb,url,hyperref,bbm}
\nolinenumbers

\title{Model Checking for Noisy Toric Words}

\author{Mihir Vahanwala}{MPI-SWS, Saarland Informatics Campus }{mvahanwa@mpi-sws.org}{https://orcid.org/0009-0008-5709-899X}{}

\authorrunning{Mihir Vahanwala}

\Copyright{Mihir Vahanwala}

\ccsdesc[500]{Theory of computation~Logic and verification}

\keywords{Model checking, dynamical systems, almost-periodic words, linear recurrence sequences, prefix-independent languages}  

\begin{document}

\maketitle
\begin{abstract}
    Toric words have recently gained attention as a means to model the behaviour of dynamical systems such as linear loops. This work is concerned with describing when the toric model is precise enough to decide whether the trace of a system is contained in an $\omega$-regular language. Our contributions are threefold. (i) We use topological means to generalise a class of toric words by defining noise-robust dynamical systems and their sets of noisy traces. We prove that all noisy traces of a system are almost-periodic, and have the same set of recurrent factors. (ii) We apply the abstract techniques above to concrete sequences obtained as sign descriptions of real algebraic linear recurrence sequences (LRS), and show that the language-membership problem for sign descriptions of LRS with few dominant roots is decidable. (iii) We show that noisy traces of noise-robust dynamical systems are indistinguishable to a prefix-independent $\omega$-regular language, and inclusion is decidable provided the common language of their recurrent factors is recursive. We incidentally obtain properties of prefix-independent $\omega$-regular languages that may be of independent interest.
\end{abstract}

\section{Introduction}

A toric word $x \in \Sigma^\omega$ is obtained by considering a collection $\{S_1, \ldots, S_k\}$ of finitely many disjoint open subsets of a torus $X$, a starting point $\alpha \in X$, and a rotation $\gamma$ such that for every $n$, the point $\alpha + n\gamma$ is contained in (the closure of) one of the open sets above. A coding $\sigma$ then associates each open set with a letter in $\Sigma$, and we define $x(n) = \sigma(S_n)$ where $\alpha + n\gamma \in \mathsf{cl}({S_n})$. Toric words are extensively studied symbolic dynamical systems (see \cite[Chap.~1]{walters2000introduction} and \cite{weyl1916gleichverteilung}). A prime example is the class of Sturmian words, where $X$ is the unit circle identified by the interval $[0, 1)$, and is partitioned as $[0, \gamma), [\gamma, 1)$. 

More recently, toric words have been identified to accurately capture the behaviour of \emph{linear loops} in program verification \cite{Karimov2023}. They also play a central role in the decidability of the monadic second-order (MSO) theory of the structure ${\langle \mathbb{N}; <, a_1^\mathbb{N}, \ldots, a_k^\mathbb{N}\rangle}$ as established by \cite{berthe2024mso}. We refer the reader to \cite{berthe2025toric} for a survey. Such works hinge on the fact that toric words, and more generally, traces of \emph{minimal compact dynamical systems}, enjoy the combinatorial property of \emph{uniform recurrence}, which is a special case of \emph{almost periodicity}. 

These properties are defined and leveraged as follows. We say that a finite word $u$ is a \emph{factor} of an infinite word $x$ if $u$ occurs in $x$; if $u$ occurs infinitely often, we say that it is a \emph{recurrent} factor of $x$. A word $x$ is recurrent when every factor $u$ is recurrent; if moreover for every $u$ there is a bound on the gaps between consecutive occurrences of $u$, then $x$ is uniformly recurrent. A word $x$ is almost-periodic if for every recurrent factor $u$, there is a bound on the gaps between its consecutive occurrences. We have that $x$ is \emph{effectively} almost-periodic if, furthermore: (i) given any index $n$, we can compute $x(n)$; (ii) given any finite word $u$, we can compute $R_u$ such that either all occurrences of $u$ are in a prefix of length $R_u$, or every factor of length $R_u$ contains an occurrence of $u$. A classic result of Sem\"enov concerning runs of deterministic automata on almost-periodic words (see e.g., \cite{semenov-english}) implies that given an $\omega$-regular language $L$ and an effectively almost-periodic word $x$, one can decide whether $x \in L$.

Indeed, throughout this paper, we characterise $\omega$-regular languages as those recognised by deterministic parity automata $\mathcal{A} = (Q, \Sigma, q_\mathsf{init}, \tau, \mathsf{col})$, where $Q$ is a set of states, $\Sigma$ is the alphabet, $q_\mathsf{init}$ is the initial state, $\tau: Q \times \Sigma \rightarrow Q$ is the transition function, and $\mathsf{col}: Q \rightarrow \mathbb{N}$ is the colouring function (note that its image is finite). The automaton $\mathcal{A}$ accepts a word $x$ if the maximum colour visited infinitely often in the run is even.

The motivation of this work lies in the observation that the model obtained by using tori to describe dynamical systems, while accurate, may still be imprecise. As an example \cite{berthe2025toric}, consider the sequence $x$ over the alphabet $\{-1, 0, 1\}$ obtained by recording the sign of the expression
$$
\sin(\varphi + n\theta) + r(n),
$$
where $\theta/\pi$ is irrational, and $r(n)$ converges to $0$. Such an infinite word arises as the sign pattern of a \emph{linear recurrence sequence (LRS}, see Sec.~\ref{sec:LRS}) with two \emph{dominant characteristic roots}. Intuitively, as $n$ grows larger, the sign pattern $x$ should bear increasing resemblance to the toric word $y$ obtained by taking $X$ to be the unit circle, $\alpha$ to be starting point with angular coordinate $\varphi$, the update $\gamma$ to be rotation by angle $\theta$, and the partition of $X$ to be the union of open semicircular arcs and their endpoints. However, the indices at which $x$ and $y$ disagree can still constitute an infinite set, and the sign pattern $x$ can in fact be provably non-toric \cite[Cor.~6.9]{berthe2025toric}.

The question that arises is: under what circumstances can the language-membership problem for the original trace be solved using the model torus?

In this example, we can in fact prove that $x$ will be almost periodic, and moreover the set of recurrent factors of $x$ will be the same as that of $y$. The dynamical system induced by the orbit of an irrational rotation on a unit circle partitioned into semicircular arcs is \emph{robust} in the sense that varying the starting point and adding diminishing \emph{noise} do not alter the almost-periodic nature of the trace, nor change the set of its recurrent factors. 

The first contribution of this paper is a generalisation of the above phenomenon. We shall use topological means to define a class of noise-robust dynamical systems and their sets of noisy traces (Def.~\ref{def:robust}). We then prove in Thm.~\ref{thm:noise-robust} that all noisy traces of a robust system are almost-periodic and have the same set of recurrent factors. Thm.~\ref{thm:eap-ingredients} then delineates when noisy traces are \emph{effectively} almost-periodic, and thus gives the ingredients to address our motivating question above.

Our second contribution is an application to LRS: we combine our results (Thm.~\ref{thm:eap-ingredients}) with the techniques of \cite{powerofpos}. In Thm.~\ref{thm:LRS}, we show that the sign pattern $x$ of a real algebraic LRS is \emph{effectively} almost periodic provided it satisfies one among a handful of spectral conditions. By a direct application of Sem\"enov's result, given such an LRS and an $\omega$-regular language $L$, we can decide whether $x \in L$. In particular, Case (2) of Thm.~\ref{thm:LRS} covers our running example. We note that this case also applies to \emph{non-simple} LRS; prior work in this setting \cite{prefixind,powerofpos} is restricted to \emph{simple} LRS where it can be shown by number-theoretic means that the sign pattern differs from the toric word in only finitely many positions.

Our third contribution is the result that noisy traces of a noise-robust dynamical system are indistinguishable to any \emph{prefix-independent} $\omega$-regular language, and inclusion is decidable provided the common language $L_\star$ of recurrent factors is recursive (Cor.~\ref{cor:modelcheck}). 
A prefix-independent language \cite{prefixind,prefixind-cpec} $L \subseteq \Sigma^\omega$ has the following defining property. Let $x, x'$ be words such that $x'$ can be obtained from $x$ by making only finitely many edits. We have that $x \in L$ if and only if $x' \in L$. Non-empty prefix-independent languages thus comprise a special class of liveness properties; we also give properties of prefix-independent $\omega$-regular languages (Props.~\ref{prop:prefix-ind-buchi}, \ref{prop:disjunctive}) that may be of independent interest. E.g., the uniform Borel probability measure assigns a prefix-independent language measure either $0$ or $1$.

Finally, we shall observe that Sturmian words do not satisfy our robustness criteria. We shall demonstrate noisy traces of Sturmian dynamical systems that are not almost-periodic, and also show that querying such traces against prefix-independent $\omega$-regular languages subsume Diophantine-hard problems.

\section{Noise-Robust Dynamical Systems}
In this section, we compile a set of properties that imply a dynamical system $(X, T: X \rightarrow X)$ is robust. As outlined in the introduction, we seek to generalise the phenomenon observed for the system defined by an irrational rotation on the $1$-torus partitioned into semicircular arcs, and the reader may use this as a running example to intuit the properties we impose and the conclusions we draw.

Throughout this paper, we shall assume that $X$ is a compact metric space with bounded metric $d$ such that every open ball $B \subseteq X$ is connected. We shall assume that the dynamical update $T: X \rightarrow X$ is a homeomorphism, and in particular, is bijective. This has the following consequences.
\begin{enumerate}
    \item Both $T$ and $T^{-1}$ are homeomorphisms, and hence for every open set $S$, the image $TS$ and pre-image $T^{-1}S$ are open sets.
    \item Both $T$ and $T^{-1}$ are uniformly continuous because $X$ is compact. In particular, we can define a modulus of continuity, i.e., a non-decreasing function $\Omega: \mathbb{R}_{\ge 0} \rightarrow  \mathbb{R}_{\ge 0}$, such that for all $\alpha, \alpha' \in X$ we have $d(T\alpha, T\alpha') < \Omega(d(\alpha, \alpha'))$, $d(T^{-1}\alpha, T^{-1}\alpha') < \Omega(d(\alpha, \alpha'))$, the function $\Omega$ is continuous at $0$, and $\Omega(0) = 0$. 
    \footnote{
    For any $\delta \ge 0$, the set $B_\delta = \{(\alpha, \alpha') \in X \times X: d(\alpha, \alpha') \le \delta\} $ 
    is closed in the compact $X \times X$, and hence compact. Take $\Omega(\delta) = 2 \cdot \max_{(\alpha, \alpha') \in B_\delta} \max(d(T\alpha, T\alpha'), d(T^{-1}\alpha, T^{-1}\alpha'))$.} 
\end{enumerate}

Our final requirement on the dynamics is that that the dynamical system $(X, T)$ be minimal, i.e., for all $\alpha \in X$ the orbit $\{T^n \alpha: n \in \mathbb{N}\}$ is dense in $X$.

\subsection{Partitions for the coding}

Having set up the dynamics, we turn to the requirements on how the coding of an orbit will be obtained. We partition $X$ into finitely many disjoint open sets $S_1, \ldots, S_k$, and a closed set $Z$ such that $\bigcup_{i=1}^k S_i$ is dense in $X$. Observe that $Z$ is compact. We impose the independent discrete orbit condition: for all $\zeta \in Z$, and all $n \ge 1$, we must have that $T^n \zeta \notin Z$. This condition is technically helpful, and is reminiscent of its counterpart which is now a staple working assumption in the study of \emph{interval exchange transformations} \cite{keane-iet}.

A coding $\sigma$ surjectively maps $X$ to a finite alphabet $\Sigma$ in an ``almost continuous way'': we have that $\sigma$ is constant on each of the open sets $S_1, \ldots, S_k$, and that if $\sigma(\zeta) = b$ for some $\zeta \in Z$, then for every $\delta > 0$, there exists $\alpha \notin Z$ such that $d(\alpha, \zeta) < \delta$ and $\sigma(\alpha) = b$.

We model noise as follows. Let $\mathcal{N} = \{\nu: \mathbb{N} \rightarrow \mathbb{R}_{\ge 0} \mid \lim_{n\rightarrow \infty} \nu(n) = 0\}$ be the set of noise functions. We say that the noise function $\nu$ \emph{converges effectively} if for every $\delta$ we can compute $N$ such that for all $n \ge N$ we have $\nu(n) < \delta$. The function $\mathsf{traces}_\sigma: X \times \mathcal{N} \rightarrow 2^{\Sigma^\omega}$ maps a starting point $\alpha$ and noise function $\nu$ to the set of infinite words 
\begin{equation}
  \mathsf{traces}_\sigma(\alpha, \nu) = \{x \in \Sigma^\omega \mid \forall n. \exists \beta. ~ d(\beta, T^n \alpha) < \nu(n) \land \sigma(\beta) = x(n)\}.  
\end{equation}

In other words, a noisy trace may code a point $\beta$ in the $\nu(n)$-ball around $T^n\alpha$ instead of $T^n \alpha$ itself. By our requirements on the coding $\sigma$, we can assume without losing generality that $\beta \notin Z$. Finally, we define 
\begin{equation}
    \mathsf{Traces}_\sigma = \bigcup_{\alpha \in X, \nu \in \mathcal{N}} \mathsf{traces}_\sigma(\alpha, \nu).
\end{equation}

We summarise our conditions in the following definition.

\begin{definition}
\label{def:robust}
    A dynamical system $(X, T, \langle S_1, \ldots, S_k, Z \rangle, \sigma)$, where $X$ is partitioned into $S_1, \ldots, S_k, Z$, and $\sigma: X \rightarrow \Sigma$ is a coding, is said to be noise-robust if:
    \begin{enumerate}
        \item The space $X$ is compact and admits a metric $d$, and moreover, every open ball $B \subseteq X$ is connected.
        \item The update $T: X \rightarrow X$ is a homeomorphism, admits a modulus of continuity $\Omega$, and induces a minimal dynamical system on $X$, i.e., for all $\alpha \in X$, the orbit $\{T^n \alpha \mid n \in \mathbb{N}\}$ is dense in $X$.
        \item The partitions $S_1, \ldots, S_k$ are disjoint open sets, and their union is dense in $X$.
        \item For all $\zeta \in Z$, and all $n \ge 1$, we have $T^n \zeta \notin Z$.
    \end{enumerate}
\end{definition}

It is straightforward to check that our running example of an irrational rotation on the $1$-torus partitioned into open semicircular arcs and their endpoints meets all the four conditions. On the other hand, toric systems defining Sturmian words satisfy all but the fourth condition: the distance between the endpoints of the intervals is the same as the angle of rotation.

\subsection{Noise Robustness}
\begin{theorem}
\label{thm:noise-robust}
  Let $(X, T, \langle S_1, \ldots, S_k, Z \rangle, \sigma)$ be a noise-robust dynamical system. We have that every $x \in \mathsf{Traces}_\sigma$ is almost-periodic. Furthermore, there exists $L_\star \subseteq \Sigma^+$ such that for every $x \in \mathsf{Traces}_\sigma$, the set of recurrent factors of $x$ is $L_\star$.
\end{theorem}
\begin{proof}
    We shall associate each word $u \in \Sigma^*$ with an open set $X_u \subseteq X$. This assignment will have the property that if the orbit falls in $X_u$, the next $|u|$ letters of the coding will be an occurrence of $u$ provided there is no noise.
    
    We assign $X_\varepsilon = X$ for the empty word, and $X_{bu} = \left(\bigcup_{i: \sigma(S_i) = b} S_i\right) \cap T^{-1}X_u$. Note that here we use the fact that $T^{-1}$ is a homeomorphism that maps open sets to open sets. We claim that $L_\star = \{u \in \Sigma^+ \mid X_u \ne \{\}\}$.

    We shall first prove that if $X_u$ is non-empty, then $u$ is a recurrent factor of all $x \in \mathsf{Traces}_\sigma$, and furthermore, there is a bound $R_u$ such that for all $x \in \mathsf{Traces}_\sigma$, the gaps between consecutive occurrences are eventually bounded by $R_u$. Since $X_u$ is an open set, we can choose a small enough $\delta$ and construct a non-empty open set $Y_u \subseteq X_u$ such that for all $\alpha \in Y_u$, the $\delta$-ball around $\alpha$ is contained in $X_u$.\footnote{Choose any $\alpha \in X_u$, use that $X_u$ is open to get a $\delta'$-ball around $\alpha$ contained in $X_u$, let $Y_u$ be the $(\delta'/3)$-ball around $\alpha$, choose $\delta = \delta'/3$, and apply the triangle inequality.} We see that if $Y_u$ is visited when the noise is less than $\delta$, then it will mark an occurrence of $u$. Since the noise converges to $0$, it would suffice to prove that the gaps between visits to $Y_u$ are bounded. 

    Since the dynamical system is minimal, we have that $\bigcup_{n=0}^\infty T^{-n}Y_u = X$, i.e., the orbit of every point in $X$ eventually visits $Y_u$. Since $Y_u$ is open and $T^{-1}$ is a homeomorphism, we get that $(T^{-n}Y_u)_{n=0}^\infty$ is an open cover of $X$, which, due to compactness, admits a finite subcover $(T^{-n}Y_u)_{n=0}^R$. Thus, every point in $X$ is either in $Y_u$ or will visit $Y_u$ in at most $R$ steps under $T$. This proves that $u$ occurs infinitely often in any noisy trace $x$, and once the noise is bounded by $\delta$, the gaps between occurrences of $u$ in the trace are bounded by $R_u = R$. 

    We now prove the converse, i.e., if $X_u$ is empty (which can happen only if the length $|u| > 1$), then $u$ can occur only finitely often in any noisy trace. We will do so by contradiction: we will show that if a noisy trace contains infinitely many occurrences of $u$, then $d(T^m Z, Z) = 0$ for some $m$ with $|m| \le |u|$, which would contradict the independent discrete orbit condition.
    
    Given $\delta$, consider a noisy trace $x \in \mathsf{traces}_\sigma(\alpha, \nu)$ such that $x(n)\cdots x(n+l-1)$ is an occurrence of $u$ at an index where the noise is guaranteed to be less than $\delta$, i.e., $\nu(n+j) < \delta$ for all $j \ge 0$. Let $\beta_0, \ldots, \beta_{l-1}$ denote $T^n \alpha, \ldots, T^{n+l-1}\alpha$, and let $\gamma_0, \ldots, \gamma_{l-1}$ denote the corresponding perturbed versions. Recall we may assume without loss of generality that $\gamma_0, \ldots, \gamma_{l-1} \notin Z$. We make the following observations.
    \begin{itemize}
        \item The points $\gamma_0, \ldots, \gamma_{l-1}$ are all correctly placed, i.e., each $\gamma_i$ falls in an open set $S_i$ that is coded with the letter $u(i)$.
        \item There does not exist any point $\beta$ such that $\beta, \ldots T^{l-1}\beta$ are all correctly placed, because the set $X_u$ is empty.
    \end{itemize}
    In particular, there exists an $i$ such that $\beta_i$ is incorrectly placed in some $S'$, while $\gamma_i$, which is in the $\delta$-ball around it, is correctly placed in some $S \ne S'$. Since the $\delta$-ball is connected (by property (1) of a robust system), it cannot be the union of disjoint open sets\footnote{The open sets in this case are the intersections of the ball with each of $S_1, \ldots, S_k$.}, and must contain some $\zeta_i \in Z$, i.e., $d(\beta_i, \zeta_i) < \delta$.

    Now, consider the points $T^{-i}\gamma_i, \ldots, T^{l-1-i}\gamma_i$. There necessarily exists some $j \ne i$ such that $T^{j-i}\gamma_i$ is incorrectly placed. By uniform continuity, we have that $d(T^{j-i}\beta_i, T^{j-i}\gamma_i) < \Omega^{|j-i|}(\delta)$, i.e., the incorrectly placed $T^{j-i}\gamma_i$ is in the $\Omega^{|j-i|}(\delta)$-ball around $\beta_j$, whereas the correctly placed $\gamma_j$ is in the $\delta$-ball around $\beta_j$. By similar connectedness arguments above, we get that the $(\max(\delta, \Omega^{|j-i|}(\delta)))$-ball around $\beta_j$, by virtue of containing both $\gamma_j$ and $T^{j-i}\gamma_i$, contains a point $\zeta_j \in Z$, or in other words, $d(\beta_j, \zeta_j) < \max(\delta, \Omega^{|j-i|}(\delta))$.

    We now use uniform continuity on $\beta_i, \zeta_i$ to argue that $d(\beta_j, T^{j-i}\zeta_i) < \Omega^{|j-i|}(\delta)$. Applying the triangle inequality to this and the result of the previous paragraph, we get that 
    $$
    d(T^{j-i}\zeta_i, \zeta_j) < \max(\delta, \Omega^{|j-i|}(\delta)) + \Omega^{|j-i|}(\delta).
    $$

    We can now supply $\delta = 1/2, 1/4, \ldots$, and hence get a sequence of $(\zeta, \zeta')$ such that $\min_{|m|\le |u|} d(T^m\zeta, \zeta')$ converges to $0$. By compactness, any limit point $(\zeta_\star, \zeta'_\star)$ of this sequence of $(\zeta, \zeta')$ lies in $Z \times Z$, is guaranteed to exist, and by continuity (of the metric $d$, homeomorphisms $T, T^{-1}$, and of the modulus of continuity $\Omega$), $$\min_{|m|\le |u|}d(T^m \zeta_\star, \zeta'_\star) = 0,$$ which contradicts the independent discrete orbit condition.

    We have thus proven that $u$ is a recurrent factor if and only if $X_u$ is non-empty, and that if $X_u$ is non-empty, the factor $u$ occurs with bounded gaps once the noise is bounded by some $\delta_u$. This establishes that for all $x \in \mathsf{Traces}_\sigma$, $x$ is almost-periodic, and the set of recurrent factors is $L_\star = \{u \mid X_u \ne \{\}\}$.
\end{proof}

We make some observations concerning effectiveness. To that end, define $\Delta: \mathbb{N} \rightarrow \mathbb{R}$ as $$\Delta(n) = \min_{|m| \le n} d(T^mZ, Z),$$ and note that $\Delta(n) > 0$ for $n > 0$. 
\begin{lemma}
    \label{lem:recurrent-factors}
    Let $L_\star \subseteq \Sigma^+$, and let $W \subseteq \Sigma^\omega$ be a set of almost-periodic words such that for every $x \in W$, the set of recurrent factors of $x$ is $L_\star$. For every $u \in L_\star$, there exists $R_u$ such that for every $x \in W$, eventually every factor of length $R_u$ contains an occurrence of $u$. Computing this $R_u$ Turing-reduces to deciding membership in $L_\star$. 
\end{lemma}
\begin{proof}
    The existence of $R_u$ follows almost immediately by definition. Indeed, consider arbitrary $x \in W$. There exists $R$ such that every length-$R$ factor of $x$ contains $u$. We can assert $R_u \le R$, because in particular every recurrent length-$R$ factor will have an occurrence of $u$. To compute $R_u$, we simply enumerate the elements of $L_\star$ until we find a length $R$ such that all length-$R$ members of $L_\star$ contain an occurrence of $u$. It remains to observe that every $x \in W$ has a suffix for which all length-$R$ factors are recurrent.
\end{proof}

\begin{theorem}
    \label{thm:eap-ingredients}
    Let $x \subseteq \Sigma^\omega$ be a noisy trace of a noise-robust dynamical system such that $L_\star$ is recursive and $\Omega, \Delta$ are computable. Let the noise function $\nu$ of $x$ converge effectively. If for each $n$, we can compute $x(n)$, then $x$ is effectively almost-periodic. Thus, given such $x \in \Sigma^\omega$ and an $\omega$-regular language $L$, we can decide whether $x \in L$.
\end{theorem}
\begin{proof}
    We need to show how, given $u \in \Sigma^+$, we can compute $R_u$ such that either $u$ can only occur in the length-$R_u$ prefix of $x$, or $u$ occurs in every length-$R_u$ factor of $x$. One can determine which of the cases holds by querying whether $u \in L_\star$.

    Suppose $u$ is not a recurrent factor. We compute $\delta$ small enough such that for all $m$ with $|m| \le |u|$, we have $\max(\delta, \Omega^m(\delta)) + \Omega^m(\delta) < \Delta(|u|)$. From the proof of Thm.~\ref{thm:noise-robust}, we deduce that there cannot be any occurrence of $u$ at indices where the noise is less than $\delta$. Since the noise function $\nu$ converges effectively, we can compute $R_u$ such that for all $n \ge R_u$, we have $\nu(n) < \delta$, and hence $u$ cannot occur at index $n$. 

    In the case $u$ is a recurrent factor, we can use the previous lemma to compute $R_u'$ such that $x$ has a suffix where all length-$R_u'$ factors are recurrent, and contain an occurrence of $u$. It remains to compute the starting index $R_u''$ of this suffix, and take $R_u = R_u' + R_u''$. To do so, we appeal to the previous case and compute a bound $R_u''$ on the index of the last occurrence of any length-$R_u'$ non-recurrent factor.

    Finally, the last statement follows from the result of Sem\"enov mentioned in the Introduction, see \cite[Chap.~3]{Karimov2023} for a detailed exposition. 
\end{proof}

We remark that in our applications, the open subsets of the torus will be semi-algebraic. This will allow us to effectively decide (see e.g., \cite[App.~A]{prefixind}), given any $u$, whether the attendant open set $X_u$ is non-empty, or in other words, whether $u \in L_\star$.

\section{Sign Descriptions of LRS}
\label{sec:LRS}
Recall that a linear recurrence sequence (LRS) of \emph{order} $d$ over a field $\mathbb{K}$ (we shall consider LRS over real algebraic numbers) is a sequence $(\mu_n)_{n=0}^\infty$ that satisfies the recurrence relation $\mu_{n+d} = a_{d-1}\mu_{n+d-1} + \cdots + a_0 \mu_n$, where $a_0 \ne 0$. We refer to the polynomial $X^d - a_{d-1}X^{d-1} - \cdots - a_0$ as the \emph{characteristic polynomial} of the LRS, and its roots are called the characteristic roots. A characteristic root that is not repeated is called \emph{simple}, an LRS is called simple if all characteristic roots are simple. An LRS for which there is no pair $\lambda_i, \lambda_j$ of distinct characteristic roots such that $\lambda_i/\lambda_j$ is a root of unity is called \emph{non-degenerate}. It is well known that any given LRS can be algorithmically decomposed as the interleaving of non-degenerate LRS (see e.g., \cite{powerofpos}). We shall assume that the distinct characteristic roots $\lambda_1, \ldots, \lambda_k$ are ordered in descending Euclidean absolute value, and shall refer to the roots with maximal absolute value as \emph{dominant}.

It is well known that LRS admit an exponential-polynomial closed form, i.e.,
$$
\mu_n = \sum_{i=1}^k \sum_{j=0}^{m_i-1} p_{ij}n^j \lambda_i^n = \sum_{i=1}^k f_i(n) \lambda_i^n.
$$
Here, we have that $m_i$ is the multiplicity of characteristic root $\lambda_i$, and can assume that $\sum_{i=1}^k m_i = d$ and  $p_{i(m_i-1)} \ne 0$.

The sign description of an LRS $(\mu_n)_{n=0}^\infty$ is the word $x \in \{1, -1, 0\}^\omega$ such that $x(n) = \text{sign}(\mu_n)$. Before we state the main result of this section, we record a lemma (see \cite[Cor.~5.5]{berthe2025toric}, \cite[Thms.~8, 9]{powerofpos}) that will be helpful in accounting for ``degeneracies.'' 

\begin{lemma}
    \label{lem:degenerate}
    Let $x_0, \ldots, x_{d-1}$ be effectively almost-periodic toric words defined by open semi-algebraic sets. The word $x$, defined as $x(qd + r) = x_r(q)$ is also toric, defined by open semi-algebraic sets, and effectively almost-periodic.
\end{lemma}

We shall also use the following self-evident observation: if we can compute $N$ such that $x(N, \infty)$ is effectively almost-periodic, then $x$ is effectively almost-periodic.

\begin{theorem}
\label{thm:LRS}
    Let $(\mu_n)_{n=0}^\infty$ be a real algebraic LRS that satisfies one of the following conditons:
    \begin{enumerate}
        \item has a single real dominant root,
        \item has two dominant roots whose ratio is not a root of unity,
        \item has three dominant roots which are all simple, and the ratio of the complex conjugate pair is not a root of unity,
        \item has order at most four.
    \end{enumerate}
    We can construct a Turing machine which witnesses that the sign description of $(\mu_n)_{n=0}^\infty$ is effectively almost-periodic. Thus, given such an LRS and an $\omega$-regular language $L$, we can decide whether the sign description $x$ is in $L$. 
\end{theorem}
\begin{proof}
    Since precise computations can be carried out on the field of real algebraic numbers, it is clear that given any $n$, we can always determine the $n$-th letter of the sign description. We shall therefore focus on showing how, given a finite word $u$, we can compute $R_u$ such that $u$ either does not occur at any index beyond $R_u$, or occurs in every length-$R_u$ factor of the sign description. 

    Case (1) is trivial. If there is a single real dominant root, then the sign description is effectively eventually periodic with period at most $2$.

    We note that showing that the sign-description is effectively almost-periodic would entail solving the Skolem problem, i.e., computing the set of indices at which the LRS is $0$. Fortunately, the cases of the statement are amenable to techniques used to solve the Skolem problem for recurrences in the MSTV class, named after Mignotte, Shorey, Tijdeman \cite{MST1984}, and Vereshchagin \cite{Vereshchagin1985}, who independently showed that the Skolem problem is decidable for LRS of order at most four. The key technical lemma uses Baker's theory of linear forms in logarithms. We refer the reader to \cite[Sec.~3.2]{bacik2025survey}, \cite[Sec.~3, 4]{bilu2025survey} for expository proofs of the following result, which makes the techniques of \cite{powerofpos} effective at low orders. 

    \begin{lemma}
    \label{lem:suffix-LRS}
        Let $(\mu_n)_{n=0}^\infty$ be a real algebraic LRS which has at most three dominant roots $\lambda_1, \ldots, \lambda_r$, and satisfies the property that the ratio of any pair of distinct dominant roots is not a root of unity. We can compute an index $N$ such that for all $n \ge N$, 
        $$
        \left|\sum_{i=1}^r f_i(n) \lambda_i^n \right| > \left|\sum_{j=r+1}^k f_j(n) \lambda_j^n \right|.
        $$
    \end{lemma}

    The lemma has the following consequences for Cases (2) and (3).
    \begin{itemize}
        \item Any word $u$ that contains the letter $0$ cannot be a factor of the sign description beyond the index $N$.
        \item In Case (3), beyond index $N$, the sign description of the given LRS matches that of $\sum_{i=1}^3 p_{i0} \lambda_i^n$, which is in turn the same as the sign pattern of a sequence of the form $a + b\cos(n\theta+\varphi)$ or $a(-1)^n + b\cos(n\theta + \varphi)$. This computable suffix of the sign description is effectively uniformly recurrent by virtue of being a toric word\footnote{The former is obviously obtained as a coding of a rotation on a $1$-torus, the latter is the interleaving of two such codings, which is still toric by Lem.~\ref{lem:degenerate}. Effectiveness follows because all computations to determine whether an open interval corresponding to a putative factor is non-empty involve algebraic numbers and semialgebraic sets.}, see, e.g., \cite[proof of Thm.~3.1]{prefixind}, \cite[Thm.~8, Thm.~9]{powerofpos}.
    \end{itemize}

The eventual effective uniform recurrence settles Case (3). We now turn to Case (2), which is tackled by generalising the arguments in \cite[Sec.~6.4]{berthe2025toric}. Beyond index $N$, the sign description of the LRS is the same as that of the sequence
$$
\sum_{j=0}^{m_i-1} a_j n^j \cos(n\theta + \varphi_j),
$$
none of whose terms are $0$. Intuitively, the sign of each factor is driven by the term $a_{m_i-1}n^{m_i-1}\cos(n\theta + \varphi_{m_i-1})$, and can be the opposite only if the cosine factor is less than $A/n$ for some effective $A$. More formally, the sign description of the above sequence is a noisy trace of the noise-robust dynamical system defined by an irrational rotation on the $1$-torus partitioned into semicircular arcs. We immediately obtain from Thm.~\ref{thm:noise-robust} that the sign description is almost-periodic. Furthermore, $L_\star$ is recursive by semi-algebraic geometry, the noise function converges effectively, the modulus of continuity $\Omega(\delta)$ is simply $\delta$, and $\Delta(n) = \min_{d, |m|\le n}|m\theta - d\pi|$ can be effectively under-approximated. We obtain from Thm.~\ref{thm:eap-ingredients} that the noisy trace, i.e., the sign description, is effectively almost-periodic.

Only Case (4) remains. Three subcases arise: (i) All four roots are distinct and dominant; (ii) There are fewer than four dominant roots, and there is no pair of distinct dominant roots whose ratio is a root of unity; (iii) There are two or three dominant roots, and the ratio between a distinct pair is a root of unity. Of these, subcase (ii) has been subsumed by prior discussions. 

In subcase (i), we have that the LRS is simple, has roots $\{\pm \rho, \rho e^{\pm i\theta}\}$ or $\{\rho e^{\pm i\theta_1}, \rho e^{\pm i\theta_2}\}$, and hence its sign pattern has a suffix that is effectively uniformly recurrent by virtue of being toric \cite[Thm.~11]{powerofpos} (more precisely, the suffix of the sign pattern is the interleaving of effectively almost-periodic toric words, and will inherit the property by Lem.~\ref{lem:degenerate}). The starting index of this suffix can be computed using Lem.~\ref{lem:suffix-LRS} or by directly solving the Skolem problem for the ensuing low-order non-degenerate LRS in the MSTV class. This identifies when all of the interleaved non-degenerate LRS are guaranteed to be non-zero.  

In subcase (iii), if there are three dominant roots, the LRS is of the form $(\pm \rho)^n + \rho^n \cos (n\pi/d + \varphi) + (\pm \gamma)^n$, where $0 < \gamma < \rho$. It is easy to see that the sign description will be effectively ultimately periodic. If there are two dominant roots, it is easy to see that the sign pattern is effectively ultimately periodic, except if the remaining non-dominant roots are complex conjugates and not roots of unity, and the contribution from dominant terms is periodically $0$. We can still apply Lem.~\ref{lem:suffix-LRS} to identify when the non-dominant contribution will persistently be non-zero, and then use Lem.~\ref{lem:degenerate} to deduce that the suffix of overall sign description, which is the interleaving of effectively almost-periodic, toric sign descriptions of low-order non-degenerate LRS, is itself toric and effectively almost-periodic.

To end the proof, the last statement follows from Sem\"enov's theorem, as in Thm.~\ref{thm:eap-ingredients}.
\end{proof}

There are known instances of order-five LRS where is Skolem problem is unresolved, and instances of LRS with three dominant roots whose sign description is provably not almost-periodic \cite[Sec.~4]{prefixind}, and instances of LRS with three dominant roots for which the positivity and ultimate positivity problems (respectively, does the sign description contain $-1$, does the sign description contain $-1$ only finitely many times) are unresolved \cite[Sec.~5]{pos-6}.

\section{Prefix-Independent Verification}
\begin{definition}
    An $\omega$-regular language $L \subseteq \Sigma^\omega$ is said to be prefix-independent if for all $u, v \in \Sigma^*$ and $x \in \Sigma^\omega$, $ux \in L$ if and only if $vx \in L$.
\end{definition}

In other words, if $x'$ is obtained from $x$ by making finitely many edits, then either $x, x' \in L$ or $x, x' \notin L$. We remark that any non-empty prefix-independent language is a liveness property, but the liveness property $(\Sigma\Sigma)^*(a\Sigma)^\omega$ is \emph{not} a prefix independent language. 

\begin{proposition}
\label{prop:prefix-ind-buchi}
    An $\omega$-regular language $L \subseteq \Sigma^\omega$ is prefix-independent if and only if $L = \bigcup_{i=1}^d \Sigma^* V_i^\omega$, where $V_1, \ldots, V_d$ are regular languages not containing the empty word.
\end{proposition}
\begin{proof}
    We assume without losing generality that $L$ is non-empty.
    
    Suppose $L = \bigcup_{i=1}^d \Sigma^* V_i^\omega$. Suppose, for the sake of contradiction, there exists a pair $x, x'$ of words a finite edit-distance apart such that $x \in \Sigma^*V_i^\omega \subseteq L$ but $x' \notin L$. We have that $x = uv_1v_2\ldots$, where for each $n \ge 1$, $v_n \in V_i$. Since $x'$ is only a finite edit-distance away from $x$, it can be factorised as $u'v_m v_{m+1}\ldots$, which witnesses that $x' \in \Sigma^*V_i^\omega$, a contradiction. Thus $L = \bigcup_{i=1}^d \Sigma^* V_i^\omega$ implies that $L$ is prefix-independent.

    Conversely, suppose $L = \bigcup_{i=1}^d U_i V_i^\omega$ is prefix-independent. Clearly, $L \subseteq \bigcup_{i=1}^d \Sigma^* V_i^\omega$. We shall prove that the reverse inclusion also holds. To that end, consider $x \in \Sigma^*V_i^\omega \subseteq \bigcup_{i=1}^d \Sigma^* V_i^\omega$. This word is clearly a finite edit-distance away from a word $x' \in U_iV_i^\omega \subseteq L$. Since $L$ is prefix-independent, this implies $x \in L$, and hence $L = \bigcup_{i=1}^d \Sigma^* V_i^\omega$.
\end{proof}

The rest of this section is organised as follows. We first formalise the intuition that prefix-independent $\omega$-regular languages comprise either ``very dense'' or ``very sparse'' subsets of $\Sigma^\omega$. This intuition is further reflected in the main result of the section: almost-periodic words with the same set of recurrent factors are indistinguishable to a prefix-independent $\omega$-regular language.

We say a word $x \in \Sigma^\omega$ is \emph{disjunctive} if the set of its recurrent factors is $\Sigma^+$. Given an alphabet $\Sigma$, we use $D_\Sigma$ to denote the set of disjunctive words over it.

\begin{proposition}
    \label{prop:disjunctive}
    Let $L \subseteq \Sigma^\omega$ be a prefix-independent $\omega$-regular language. We have that either $D_\Sigma \subseteq L$ or $D_\Sigma \cap L = \{\}$. Furthermore, we can determine which is the case.
\end{proposition}
\begin{proof}
    Consider a deterministic parity automaton $\mathcal{A}$ recognising $L$, and in particular, the graph induced by it. We can prove \cite[Thm.~4.16]{berthe2024mso} (detailed proof in \cite[App.~A.4]{berthe2024msoarxiv}) that in the run on any disjunctive word $x$, the set of states visited infinitely often is precisely a bottom strongly connected component. To determine whether the set of disjunctive words is included in, or excluded from $L$, we check whether the highest colour in a bottom strongly component is even or odd. Note that the parity must be consistent across different bottom strongly components, otherwise the prefix-independence of $L$ would be contradicted. (One can make finitely many edits to a disjunctive word so that the run settles in any chosen bottom strongly connected component; these edits, however, alter neither disjunctivity nor membership in the language.) 
\end{proof}

In other words, a prefix-independent $\omega$-regular language either contains almost all words, or excludes almost all words. Here, the quantifier ``almost all'' is with respect to the uniform Borel probability measure on $\Sigma^\omega$, which assigns the cylinder $\{ux \mid x \in \Sigma^\omega\}$ measure $1/|\Sigma|^{|u|}$. 

The main result of this section is a slightly generalised reformulation of \cite[Sec.~5]{prefixind}, with simplified techniques. (We leave gauging the extent of simplification to the discernment of the reader.) We introduce an auxiliary technical notion before we prove the theorem. 

In order to aggregate information about automaton runs, we define a \emph{journey} to be a tuple in $Q \times Q \times \mathbb{N}$, which describes a finite run $q_0q_1\ldots q_l$ on a word $u_0\ldots u_{l-1}$ by recording the starting state $q_0$, the ending state $q_l$, and the maximum colour among those of $q_1, \ldots, q_l$. Clearly, if a word $u$ can make the journey $(q_1, q_2, c_1)$ and a word $v$ can make the journey $(q_2, q_3, c_2)$, then the word $uv$ can make the journey $(q_1, q_3, \max(c_1, c_2))$. For a factorisation into finite words $u_0u_1\cdots$ of an infinite word $x$, if each $u_i$ makes can make the journey $(q_i, q_{i+1}, c_i)$, then the automaton accepts $x$ if and only if $\limsup_{i \in \mathbb{N}}c_i$ is even.

\begin{theorem}
\label{thm:prefix-ind}
    Let $W \subseteq \Sigma^\omega$ be a set of almost-periodic words, and let $L_\star \subseteq \Sigma^*$ be such that for all $x \in W$, the set of recurrent factors of $x$ is $L_\star$. Let $L \subseteq \Sigma^\omega$ be a prefix-independent $\omega$-regular language. We have that $W \subseteq L$ or $W \cap L = \{\}$. Furthermore, given an oracle for membership in $L_\star$, we can decide which of the above two cases holds.
\end{theorem}
\begin{proof}
    We will show that for any $x \in W$, we can decide whether $x \in L$ purely with oracle queries to $L_\star$. Throughout this proof, we assume that $L$ is recognised by a deterministic parity automaton $\mathcal{A} = (Q, \Sigma, q_\mathsf{init}, \Delta, \mathsf{col})$, and all states in $\mathcal{A}$ are reachable from $q_\mathsf{init}$.

   Observe, using Lem.~\ref{lem:recurrent-factors}, that given $u \in L_\star$, any arbitrary $x \in W$ can be written as $w u v_1 u v_2 u \cdots$, where $|v_i| < R_u$ and $uv_iu \in L_\star$ for all $i \ge 1$. Indeed, we choose $w$ long enough so that in the suffix obtained by deleting it, all factors of length at most $2|u|+R_u$ are recurrent, and all factors of length $R_u$ contain an occurrence of $u$. We remark that $R_u$ is effective, but $w$ need not be effective for the proof that follows.
   
   We define the set $\mathsf{window}_u = \{v \mid |v| < R_u \land uvu \in L_\star \}$ and observe that it is finite and computable with oracle access to $L_\star$. In other words, $W \subseteq \Sigma^* (u \cdot \mathsf{window}_u)^\omega$. We will use this pattern to determine the highest colour seen infinitely often in a run of $\mathcal{A}$.

    We say that a factor $u \in L_\star$ is \emph{saturated} with respect to $\mathcal{A}$ if for all $q \in Q$ and all $v \in \mathsf{window}_u$, if $u$ makes the journey $(q, q', c)$ and $uvu$ makes the journey $(q, q'', c')$, then $c = c'$. In other words, the highest colour seen while reading $uvu$ starting from any $q$ must have been already seen while reading the prefix $u$. Saturated factors help us determine acceptance as follows. Let $x = w u v_1 u v_2 \cdots$, and let the corresponding sequence of journeys be $(q_0, q_1, c_0), (q_1, q_2, c_1), \ldots$. By saturation, for all $i$, we get that $c_{2i+1} \ge c_{2i+2}$ and $c_{2i+1} \ge c_{2i+3}$. We now use prefix-independence.
    
    Let $(q_\star, q', c_\star)$ be a journey $u$ can undertake such that for all other journeys $(q, q'', c)$ of $u$, $c \ge c_\star$. Let $w'$ be a word that takes $q_\mathsf{init}$ to $q_\star$. By prefix-independence, we have that $x \in L$ if and only if $x'= w' u v_1 u v_2 \cdots \in L$. The new sequence of journeys is $(q_0', q_1', c), (q_1', q_2', c_1'), \ldots$, where $q_1' = q_\star$ and $c_1' = c_\star$. The same saturation properties hold: $c'_{2i+1} \ge c'_{2i+2}$ and $c'_{2i+1} \ge c'_{2i+3}$. However, $c'_{2i+1} \ge c_\star$, and $c'_1 = c_\star$. Hence, $\limsup_i c_i' = c_\star$.

    We have just proven that an arbitrary $x \in W$ is accepted by $\mathcal{A}$ if and only if the colour $c_\star$ effectively defined by a saturated factor $u \in L_\star$ is even. It remains to prove that saturated factors exist, and can be computed with oracle access to $L_\star$.

    We shall do so with a fix-point algorithm that returns a saturated factor. We associate with each word $u$ a function $f_u: Q \rightarrow \mathbb{N}$, defined such that $u$ makes the journey $(q, q', f_u(q))$ for some $q'$. We define a quasi-ordering on finite words and say that $u \preceq v$ if $f_u(q) \le f_v(q)$ for all $q$. Clearly, if $u$ is a prefix of $v$ then $u \preceq v$. In particular, for all $v \in \mathsf{window}_u$, $u \preceq uvu$, but for a saturated word $u$, $uvu \preceq u$ for all $v \in \mathsf{window}_u$. It is also easy to observe that this quasi-ordering has finitely many (at most $|Q|^2$) equivalence classes, and hence cannot admit a strictly increasing chain of more than $|Q|^2$ words. Our algorithm thus starts with $u \in L_\star$. In each iteration, it computes $\mathsf{window}_u$ with queries to $L_\star$, and checks if $u$ is saturated. If yes, it returns $u$; if not, it iterates with $uvu, v\in \mathsf{window}_u$ that is strictly above $u$ in the quasi-ordering. By the argument above, the algorithm runs for at most $|Q|^2$ iterations.
\end{proof}

Applying Thms.~\ref{thm:noise-robust} and \ref{thm:prefix-ind} gives the following result.
\begin{corollary}
\label{cor:modelcheck}
    Let $(X, T, \langle S_1, \ldots, S_k, Z\rangle, \sigma)$ be a noise-robust dynamical system, and let $L \subset \Sigma^\omega$ be a prefix-independent $\omega$-regular language. We have that either $\mathsf{Traces}_\sigma \subseteq L$ or $\mathsf{Traces}_\sigma \cap L = \{\}$. Furthermore, if $L_\star$ is recursive, the language of recurrent factors of words in $\mathsf{Traces}_\sigma$, then we can decide which of the above two cases holds.
\end{corollary}

\section{The Sturmian Case}
For Sturmian systems, the space $X$ is the $1$-torus (identified by the interval $[0, 1)$), and the dynamical update $T$ is an irrational rotation $\gamma$, i.e., $\alpha$ is mapped to $\alpha + \gamma$ mod $1$. Without loss of generality, we assume $\gamma < 1/2$. The partitions are $[0, \gamma)$ and $[\gamma, 1)$, and the coding maps these to $1$ and $0$ respectively. We see that $Z = \{0, \gamma\}$, and $T \circ 0 = \gamma$. Hence, Sturmian systems are not noise-robust by our definition. More specifically, we can prove that if $u \in L_\star$, the set of words for which $X_u$ is non-empty, then $u$ is a recurrent factor of all (noisy) traces. However, if $u \notin L_\star$, the word $u$ can still be a factor of some noisy traces, and in fact these noisy traces can fail to be almost-periodic. Indeed, the independent discrete orbit property needed to preclude this failure does not hold.

In the absence of noise, it is known that Sturmian words are uniformly recurrent and have minimal factor complexity: there are exactly $n+1$ factors of length $n$. In particular, the length-$2$ factors are $00, 01, 10$. We shall show that if we allow noise, regardless of how rapidly the noise converges to $0$, there are starting points $\alpha$ such that the noisy trace has $11$ as a recurrent factor, and moreover the gaps between the occurrences of $11$ are unbounded.

Consider a noise function $\nu$ such that for all $n$, $\nu(n) > 0$. Observe that $11$ can only occur as a factor if we perform a perturbation when $$\alpha + n\gamma \in (1- \nu(n), 1) \cup [0, \nu(n+1)),$$ or by an appropriate shifting of coordinates,
$$
\alpha \in (-n\gamma - \nu(n), -n\gamma + \nu(n+1)).
$$

To define where to perform perturbations, we shall now construct a sequence of indices $n_0, n_1, \ldots$ having unbounded gaps (e.g., by choosing $n_{i+1} > n_i + i + 1$), and a sequence $A_0 \supseteq A_1 \supseteq \ldots$ of non-empty closed intervals of the $1$-torus, such that $$A_i \subseteq (-n_i\gamma - \nu(n_i), -n_i\gamma + \nu(n_i + 1)).$$

By the nested intervals theorem, we would then have that the intersection of these intervals contains a point $\alpha_\star$, and for all $n_i$,
$$
\alpha_\star \in (-n_i\gamma - \nu(n_i), -n_i\gamma + \nu(n_i+1)).
$$
This would imply that $\alpha_\star$ has a noisy trace $x$, obtained by performing perturbations at each of the $n_i$'s (or $(n_i+1)$'s, as required), such that $x$ has $11$ as a recurrent factor, and the gaps between consecutive occurrences of $11$ are unbounded.

We now present the inductive construction of $n_0, n_1, \ldots$, and $A_0, A_1, \ldots$. We can choose $n_0$ arbitrarily, and take $A_0$ to be a closed interval, $$A_0 \subseteq (-n_0\gamma -\nu(n_0), -n_0\gamma + \nu(n_i+1)).$$

For the inductive step, assume that $n_i, A_i$ have been constructed. Let $B_i \subseteq A_i$ be an open interval. By the (Kronecker/Weyl) equidistribution theorem, we know that there are infinitely many $n$ such that $-n\gamma \in B_i$. We choose $n_{i+1}$ such that $n_{i+1} - n_i > i+1$, and choose $A_{i+1}$ to be a closed interval in the non-empty intersection of two open intervals,
$$
A_{i+1} \subseteq B_i \cap (-n_i\gamma - \nu(n_i), -n_i\gamma + \nu(n_i+1)).
$$
The infinite intersection of the $A_i$'s is thus $\{\alpha_\star\}$, where $\mathsf{traces}_\sigma(\alpha_\star)$ contains a word that is not almost-periodic, and whose set of recurrent factors contains $11$, which is not a factor of any trace with noise $0$.

Interestingly, however, if we were to change the Sturmian system by making the rotation $m\gamma$ instead of $\gamma$ (where $m$ is an integer, $|m| > 1$), then the new dynamical system would indeed be noise-robust. This shows that a noisy Sturmian word, while itself not almost-periodic, is an interleaving of almost-periodic noisy toric words. This is an analogue of \cite[Cor.~6.9]{berthe2025toric}

The same arguments as above would allow us to encode the computations of Diophantine approximation constants (e.g., $\inf n[n\gamma + \alpha]$, $\liminf n[n\gamma+ \alpha]$, where even algorithmically determining whether these quantities are $0$ remains an open number-theoretic problem) using queries such as, ``given $\alpha$ and noise $\nu$ lower-bounded by $1/n$, does $\mathsf{traces}_\sigma(\alpha, \nu)$ contain a trace $x$ in which the factor $11$ occurs (respectively, occurs infinitely often)?''

\bibliographystyle{plainurl}
\bibliography{main}

\end{document}